\documentclass[a4paper]{article}

\usepackage{amsmath}
\usepackage{amsthm}
\usepackage{amsxtra}
\usepackage{amssymb}

\usepackage{microtype}
\usepackage{verbatim} 
\usepackage{hyperref}
\usepackage{url}
\usepackage[utf8]{inputenc}

\newtheorem{definition}{Definition}
\newtheorem{lemma}[definition]{Lemma}

\newtheorem{corollary}[definition]{Corollary}
\newtheorem{theorem}[definition]{Theorem}

\newcommand{\Gon}{\ensuremath{G_{\text{on}}}}
\newcommand{\Von}{\ensuremath{V_{\text{on}}}}
\newcommand{\Voff}{\ensuremath{V_{\text{off}}}}
\newcommand{\Eon}{\ensuremath{E_{\text{on}}}}
\newcommand{\Eoff}{\ensuremath{E_{\text{off}}}}
\newcommand{\non}{\ensuremath{n_{\text{on}}}}
\newcommand{\noff}{\ensuremath{n_{\text{off}}}}
\newcommand{\mon}{\ensuremath{m_{\text{on}}}}
\newcommand{\moff}{\ensuremath{m_{\text{off}}}}
\newcommand{\supergraph}{increment graph}

\title{Incremental and Fully Dynamic Subgraph Connectivity For Emergency Planning}

\author{Monika Henzinger \\
University of Vienna, Faculty of Computer Science, Vienna, Austria\\
  \texttt{monika.henzinger@univie.ac.at}
\and
  Stefan Neumann \\
  University of Vienna, Faculty of Computer Science, Vienna, Austria\\
  \texttt{stefan.neumann@univie.ac.at}}

\begin{document}
\maketitle

\begin{abstract}
During the last 10 years it has become popular to study dynamic graph problems
in a \emph{emergency planning} or \emph{sensitivity} setting:
Instead of considering the general fully dynamic problem,
we only have to process a \emph{single} batch update of size $d$;
after the update we have to answer queries.

In this paper, we consider the dynamic subgraph connectivity problem with sensitivity $d$:
We are given a graph of which some vertices are activated and some are deactivated.
After that we get a single update in which the states of up to $d$ vertices are changed.
Then we get a sequence of connectivity queries in the subgraph of activated vertices.

We present the first fully dynamic algorithm for this problem
which has an update and query time
only slightly worse than the best decremental algorithm.
In addition, we present the first incremental algorithm which is tight with respect
to the best known conditional lower bound; moreover, the algorithm is simple and we believe
it is implementable and efficient in practice.
\end{abstract}

\section{Introduction}
Dynamic graph algorithms maintain a data structure to answer queries
about certain properties of the graph
while the underlying graph is changed, e.g., by vertex or edge deletions and additions;
such properties could be, for example, the connectivity or the shortest
paths between two vertices.
The main goal is that after an update the algorithm does not have to recompute
the data structure from scratch, but only has to make a
small number of changes to it.
Due to strong conditional lower bounds for various dynamic graph
problems (see~\cite{abboud2014popular,henzinger2015unifying,kopelowitz2016higher}),
it is necessary to restrict the dynamic model in some way to improve the efficiency
of the operations. One model that has become increasingly popular is
to study dynamic graph problems in a \emph{sensitivity} or
\emph{emergency planning} setting (see, e.g.,
\cite{duan2009dual,patrascu2007planning,chechik2011sensitivity,demetrescu2008oracles,
bernstein2009nearly,bernstein2008improved,khanna2010approximate}):
Instead of considering the general fully dynamic problem in which we get
a sequence of updates and queries, we only allow for a single batch
update of size $d$ after which we want to answer queries.
Since we allow only a single update, the update and query times for such
sensitivity problems are much faster than for the general fully
dynamic problem.

In this paper, we consider the subgraph connectivity problem with sensitivity $d$:
We get a graph $G = (V,E)$ of which some vertices are \emph{activated} and
some are \emph{deactivated} and we can preprocess it.
There is a single update changing the states of up to $d$ vertices.
In the subsequent queries we need to answer if two given
vertices are connected by a path which traverses only activated vertices.
If the update can only active previously deactivated vertices, then
an algorithm for this problem is called \emph{incremental};
if it can only deactivate activated vertices, then it is \emph{decremental};
if it can turn vertices on and off arbitrarily, then it is called \emph{fully dynamic}.

The problem is of high practical interest as it models a scenario which is
very relevant to infrastructure problems.
For example, assume you are an internet service provider and you maintain
many hubs which are connected to each other.
In case of a defect, some of the hubs fail but there is a
small number of \emph{backup} hubs which can be used until
the defect hubs are repaired in order to provide your services to your
customers. Notice that in such a scenario it is likely that the number
of backup hubs is much smaller than the number of regular hubs.

\subsection{Our Contributions}
\label{subsec:contributions}
We present the first incremental and fully dynamic
algorithms for the subgraph connectivitiy problem with
sensitivity $d$. 
The update and query times of our fully dynamic algorithm are only
slightly slower than those of the best decremental algorithm for this
problem.
In addition, the incremental algorithm is essentially tight with
respect to the best known conditional lower bound for this problem.
Additionally, we contribute a characterization of the paths which
are added to a graph when activating some nodes.

Our result for the fully dynamic problem with sensitivity $d$ is given in the following
theorem. We state the running time with respect to a blackbox algorithm
for the decremental version of the problem as subprocedure.
The number of initially deactivated vertices is denoted by $\noff$.

\begin{theorem}
\label{Thm:FullyDynamicAlgo}
  Assume there exists an algorithm for the decremental subgraph connectivity
  problem with sensitivity $d$ that has preprocessing time $t_p$, update
  time $t_u$, query time $t_q$ and uses space $S$.
  Then there exists an algorithm for the fully dynamic subgraph connectivity
  problem with sensitivity $d$ that uses space
  $O(\noff^2 \cdot S)$ and has preprocessing time $O(\noff^2 \cdot t_p)$.
  It can process an update of $d$ vertices in time $O(d^2 \cdot \max\{t_u,t_q\})$ and queries
  in time $O(d \cdot t_q)$.
\end{theorem}

For the decremental version of the subgraph connectivity problem with sensitivity $d$
(which is also referred to as \emph{d-failure connectivity}),
the best known algorithm is by Duan and Pettie~\cite{duan2010connectivity}.
Their result is given in the following lemma.

\begin{lemma}[\cite{duan2010connectivity}]
\label{Lem:DecAlgo}
	Let $G = (V,E)$ be a graph and let $n = |V|$, $m = |E|$,
	let $c \in \mathbb{N}$. Then there exists a data structure
	for the decremental subgraph connectivity problem with sensitivity~$d$ that
	has size
	$S = O(d^{1 - 2/c} m n^{1/c - 1/(c \log(2d))} \log^2 n)$
	and preprocessing time $\tilde O(S)$.
	An update deactivating $d$ vertices 
	takes time $O(d^{2c+4} \log^2 n \log \log n)$ and subsequent connectivity
	queries in the graph after the vertex deactivations take $O(d)$ time.
\end{lemma}

As pointed out in~\cite{duan2010connectivity}, for moderate values of $d$ the space $S$ used
by the data structure from Lemma~\ref{Lem:DecAlgo} is $o(m n^{1/c})$;
further, if $m < n^2$, then we always have $S = o(m n^{2/c})$.
Using the algorithm of Lemma~\ref{Lem:DecAlgo} as a subprocedure
for our result from Theorem~\ref{Thm:FullyDynamicAlgo}, we obtain the following
corollary. The number of initially activated vertices is given by $\non$
and the number of initially activated edges in the graph is denoted by $\mon$.

\begin{corollary}
\label{Cor:FullyDynamicAlgo}
  There exists an algorithm for the fully dynamic subgraph connectivity
  problem with sensitivity $d$ with the following properties.
  For any $c \in \mathbb{N}$, it uses space
  $S' = O(\noff^2 \cdot S)$ and preprocessing time $\tilde O(S')$,
  where $S = O(d^{1-2/c} \mon \non^{1/c - 1/(c \log(2d))} \log^2 \non)$.
  It can process an update of $d$
  vertices in time $O(d^{2c+6} \log^2 \non \log\log \non)$ and answer queries
  in the updated graph in time $O(d^2)$.
\end{corollary}

In the case that we get an update of size $d' < d$, we can make the update and query times
of the data structure depend only on $d'$: We build the data structure for all values
$d' = 2^1, \dots, 2^\ell$, where $\ell$ is the smallest integer such that
$d \leq 2^\ell$. Asymptotically this will not use more space than building
the data structure once for $d$;
for an update of size $d'$ we use
the instance of the data structure for the smallest $2^i \geq d'$.

In the incremental algorithm we only allow for initially
deactivated vertices to be activated.
Our result for the incremental problem is given in the following theorem.

\begin{theorem}
\label{Thm:IncAlgo}
  There exists an algorithm for the incremental subgraph connectivity
  problem with sensitivity $d$ which has preprocessing time $O(\noff^2 \cdot \non + m)$,
  update time $O(d^2)$ and query time $O(d)$. It uses space $O(\noff \cdot n)$.
\end{theorem}

The algorithm is simple and we believe it is implementable
and efficient in practice.

For our incremental data structure the sensitivity
parameter~$d$ does not have to be fixed beforehand, i.e., once initialized, the
data structure can process updates of arbitrary sizes and the update and query times
will only depend on the size of the given update.

We observe that the conditional lower bound given in
Henzinger et al.~\cite{henzinger2015unifying} for the decremental version of the problem
can easily be altered to work for the incremental problem as well.
The conditional lower bound states that
under the Online Matrix vector~(OMv) conjecture
any algorithm solving the incremental subgraph connectivity problem
with sensitivity $d$ which uses preprocessing time polynomial in $n$ and
and update time polynomial in $d$ must have a query time of $\Omega(d^{1-\varepsilon})$
for all $\varepsilon > 0$. Examining the proof of the lower bound, we observed that
the maximum of the query and update time even has to be in $\Omega(d^{2-\varepsilon})$ for all
$\varepsilon > 0$.
Hence, the update and query times of our incremental algorithm are essentially optimal
under the OMv conjecture.

\subsection{Related Work}
In recent years there have been several results studying data structures
for problems in an emergency planning or sensitivity setting when only a single
update of small size is allowed.
The field was introduced by Patrascu and Thorup~\cite{patrascu2007planning}
who considered connectivity queries after $d$ edge failures.
Demetrescu et al.~\cite{demetrescu2008oracles} studied distance oracles avoiding
a single failed node or edge.
This setting was also considered by Bernstein and Karger~\cite{bernstein2008improved,bernstein2009nearly}.
Later, Duan and Pettie~\cite{duan2009dual} studied distance and connectivity oracles in case
of two vertex failures.
Khanna and Baswana~\cite{khanna2010approximate} studied approximate shortest paths
for a single vertex failure.
As mentioned in Section~\ref{subsec:contributions}, Duan and Pettie~\cite{duan2010connectivity} studied the
decremental subgraph connectivity problem with sensitivity~$d$.
Chechik et al.~\cite{chechik2011sensitivity} considered distance oracles and routing schemes in
case of $d$ edge failures.

For the decremental subgraph connectivity problem with sensitivity $d$
there also exist conditional lower bounds
by Henzinger et al.~\cite{henzinger2015unifying} from the OMv conjecture
and most recently by Kopelowitz, Pettie and Porat~\cite{kopelowitz2016higher}
from the 3SUM conjecture.
The highest conditional lower bounds is the one
in~\cite{henzinger2015unifying}, which states that
under the OMv conjecture any algorithm using preprocessing time polynomial in $n$ and
and update time polynomial in $d$ must have a query time of $\Omega(d^{1-\varepsilon})$
for all $\varepsilon > 0$. Hence, the query time of
the decremental algorithm by Duan and Pettie~\cite{duan2010connectivity} is essentially optimal
with respect to the lower bound.

The general subgraph connectivity problem, which allows for an arbitrary number
of updates, has gained an increasing interest during the last years.
The problem was introduced by Frigoni and Italiano~\cite{frigoni2000dynamically},
who studied it for planar graphs; they achieved amortized polylogarithmic update and
query times.
In general graphs, Duan~\cite{duan2010new} constructed a data structure
which uses almost linear space, preprocessing
time $\tilde O(m^{6/5})$, worst-case update time $\tilde O(m^{4/5})$ and
worst-case query time $\tilde O(m^{1/5})$.
In an amortized setting, the data structure given by Chan, Patrascu and Roditty~\cite{chan2011dynamic}
has an update time $\tilde O(m^{2/3})$ and query time $\tilde O(m^{1/3})$; its space usage and
preprocessing time is $\tilde O(m^{4/3})$. This improved an earlier result by
Chan~\cite{chan2002dynamic} significantly. The data structure of~\cite{chan2011dynamic}
was later improved by Duan~\cite{duan2010new} to use only $\tilde O(m)$ space.
Baswana et al.~\cite{baswana2015dynamic} gave a deterministic worst-case
algorithm with update time $\tilde O(\sqrt{mn})$ and query time $O(1)$.
Further, conditional lower bounds were derived for the subgraph connectivity
problem from multiple conjectures~\cite{abboud2014popular,henzinger2015unifying}.
The highest such lower bound was given in~\cite{henzinger2015unifying};
it states that under the OMv conjecture,
the subgraph connectivity problem cannot be solved faster than
with update time $\Omega(m^{1-\delta})$ and query time $\Omega(m^\delta)$
for any $\delta \in (0,1)$ when we only allow polynomial preprocessing time
of the input graph.
Hence, the update and query times of the aforementioned algorithms are optimal up to
polylogarithmic factors and tradeoffs between update and query times.

Compared to the subgraph connectivity problem, it has a much longer tradition to
study the (edge) connectivity problem in which updates delete or add edges to the graph.
Henzinger and King~\cite{henzinger1999randomized} were the first to give an algorithm
with expected polylogarithmic update and query times;
the best algorithm using Las Vegas randomization is by Thorup~\cite{thorup2000near}
with an amortized update time of $O(\log n (\log \log n)^3)$.
Holm, de Lichtenberg and Thorup~\cite{holm2001poly} gave the first deterministic
algorithm with amortized polylogarithmic update times; currently the best such
algorithm is given by Wulff-Nilsen~\cite{wulff2013faster} which has an update time
of $O(\log^2 n / \log \log n)$.
Recently, Kapron, King and Mountjoy~\cite{kapron2013dynamic} were able to
provide the first data structure which has expected \emph{worst case} polylogarithmic time
per update and query. The result of~\cite{kapron2013dynamic}
was lately improved by Gibb et al.~\cite{gibb2015dynamic}
to have update time $O(\log^4 n)$.
However, the best deterministic worst case data structures
still have running times polynomial in the number of nodes of the graph.
For a long time the results by Frederickson~\cite{frederickson1983data}
and Eppstein et al.~\cite{eppstein1997sparsification}
running in time $O(\sqrt{n})$ were the best known. Only recently
this was slightly improved by Kejlberg-Rasmussen et al.~\cite{rasmussen2015faster},
who were able to obtain a worst case update time of
$O\left(\sqrt{\frac{n (\log \log n)^2}{\log n}}\right)$.

The rest of the paper is outlined as follows:
We start with notation and preliminaries in Section~\ref{Sec:prelim}.
In Section~\ref{Sec:IncAlgo} we prove the results for the incremental
algorithm which will already contain the main ideas for
the more complicated fully dynamic algorithm.
Section~\ref{Sec:FullyDynamicAlgo} provides the main result of this paper.

\section{Preliminaries}
\label{Sec:prelim}

In this section, we formally introduce the subgraph connectivity problem with
sensitivity $d$. At the end of the section, we show a
lemma that characterizes when disconnected vertices become connected
after activating additional vertices; the lemma will be essential to
prove the correctness of our algorithms.

The \emph{subgraph connectivity problem} with sensitivity $d$
is as follows:
Let $G = (V,E)$ be a graph with $n$ vertices and $m$ edges and
a partition of the vertices into sets \Von{} and \Voff.
The vertices in \Von{} are said to be \emph{turned on}
or \emph{activated} and those in $\Voff$ are said
to be \emph{turned off} or \emph{deactivated}.
We get a single batch update in which the states of up to $d$ vertices are be changed.
In a query for two vertices $u$ and $v$,
the algorithm has to return if there exists a path from $u$ to $v$ only
traversing activated vertices.

As we consider the subgraph connectivity problem in a sensitivity setting,
after processing a single update and a sequence of queries, we roll back
to the initial input graph. Hence, the data structure does not allow to alter
the graph by an arbitrary amount. This allows us to offer much faster update
and query times than the best algorithms which solve the general
fully dynamic problem.

We introduce more notation.
By $\Gon$ we denote the projection of $G$ on the vertices which are initially on,
i.e., $\Gon = G[\Von] = (\Von, \Eon)$, where $\Eon = \{ (u,v) \in E : u,v \in \Von \}$.
We set $\Eoff = E \setminus \Eon$ to the set of edges which have at least
one endpoint in \Voff.
To distinguish between the sizes of the activated and deactivated vertices and edges,
we set $\non = |\Von|$ to the number of activated vertices and $\noff = |\Voff|$ to the
number of deactivated vertices. Further, we set $\mon = |\Eon|$ to the number of edges
in \Gon{} and $\moff = |\Eoff|$.

With this notation we can quickly describe the main difficulties of the
subgraph connectivity problem: If \Gon{} is connected, then already
deactivating a single vertex of \Von{} can make it fall apart into
$\Theta(\non)$ connected components; on the other hand, in \Gon{} we can have
$\Theta(\non)$ connected components initially and activating a single vertex of \Voff{}
with $\Theta(\non)$ edges can make the resulting graph connected.
Hence, when deactivating or activating vertices, the number of connected
components can change arbitrarily much. However, the update and query
times of our algorithms are not supposed to polynomially depend on $n$,
but only on the size of the udpate $d$ which will usually be much smaller.

\subsection{Characterisation of Paths After Activating Vertices}
In this subsection, we introduce the terminology to characterize
when vertices in a graph $G$ become connected after we activated
the vertices of a set $I$.

We say that a deactivated vertex $v \in \Voff$ 
and a connected component $C$ of $G$ are \emph{adjacent},
if there exists a vertex $u \in C$ such that $(u,v) \in E$.
Two vertices $u,v \in \Voff$ are \emph{connected via a connected component},
if (1) there exists a connected component $C$ to which
both $u$ and $v$ are adjacent or (2) if $(u,v) \in E$.
In other words, $u$ and $v$ are connected via a connected component if they
can reach each other by a path that only traverses vertices from a single connected
component of $G$ or if $u$ and $v$ are connected by an edge.
Two connected components $C_1 \neq C_2$ are
\emph{connected by the set $I$} if there exists
a sequence of vertices $v_1, \dots, v_k \in I$ such that
(1) $v_1$ is adjacent to $C_u$,
(2) $v_k$ is adjacent to $C_v$ and
(3) $v_i$ and $v_{i+1}$ are connected via a connected component
for all $i = 1, \dots, k-1$.

We can characterize when two disconnected vertices become connected in $G$
after the vertices of the set $I$ are activated. This is done in the following lemma.

\begin{lemma}
\label{lem:paths}
  Let $G = (V,E)$ be a graph with \Von{} and \Voff{} as before.
  Further, let $I \subseteq \Voff{}$ be a set of vertices which is activated.
  Let $u, v$ be two disconnected vertices in $\Gon$ and let
  $C_u \neq C_v$ be their connected components.
  Then $u$ and $v$ are connected in $G$ after activating the vertices in $I$
  if and only if
  $C_u$ and $C_v$ are connected by the set $I$.
\end{lemma}
\begin{proof}
  Assume $u$ and $v$ are connected in $G$ after activating the vertices in $I$.
  Then there exists a path $u = w_0 \to w_1 \to \cdots \to w_\ell \to w_{\ell + 1} = v$ in $G$;
  let $w_{j_1}, \dots, w_{j_r}$ be the vertices of the path which are from the set $I$
  with $j_i < j_{i+1}$ for all $i=1,\dots,r$.
  Now observe that for all $i$, the vertices $w_{j_i + 1}, \dots, w_{j_{i+1} - 1}$
  must be in the same connected
  component $C_{j_i}$. Clearly, $w_{j_i}$ and $w_{j_{i+1}}$ are adjacent to $C_{j_i}$ and they are
  connected by the connected component $C_{j_i}$.
  The same arguments can be used to show that $C_u$ and $w_{j_1}$ are adjacent and to show that
  $C_v$ and $w_{j_r}$ are adjacent. This implies that $C_u$ and $C_v$ are connected by the set $I$.

  The other direction of the proof is symmetric.
\end{proof}

We will use Lemma~\ref{lem:paths} to argue about the correctness of our algorithms.
In particular, when we prove the correctness of our algorithms we show
that the connected components of the query vertices become connected by the set $I$
of newly activated vertices.
This is useful as we can preprocess which vertices of \Voff{} are
connected via connected components and which deactivated vertices are
reachable from the connected components of $G$.
With these properties, we are able to avoid having to keep track
of all connected components of $G$ after an update.

\section{Incremental Algorithm}
\label{Sec:IncAlgo}

In this section, we describe an algorithm for the incremental subgraph connectivity problem
that has preprocessing time $O(\noff^2 \cdot \non + m)$, update time $O(d^2)$, query time
$O(d)$ and uses space $O(\noff \cdot n)$. This will prove Theorem~\ref{Thm:IncAlgo} stated
in the introduction.

The main idea of the algorithm is to exploit Lemma~\ref{lem:paths} by preprocessing
which deactivated vertices are connected by connected components of \Gon{} and
preprocessing the adjacency of deactivated vertices and connected components of \Gon{}.

\subsection{Preprocessing}
We first compute the connected components $C_1, \dots, C_k$ of $\Gon$ and
label each vertex in \Von{} with its connected component.
For each connected component $C_i$, we use a binary array $A_{C_i}$ of size $\noff$
to store which vertices in $\Voff$ are adjacent to $C_i$.
We further equip each vertex $u \in \Voff$ with a binary array $A_u$ of size
$k + \noff = O(n)$:
In the first $k$ entries of $A_u$, we store to which connected components
$u$ is connected;
in the final $\noff$ components of $A_u$, we store to
which $v \in \Voff$ the vertex $u$ is connected by a connected component.

When the algorithm performs updates, it uses the arrays $A_u$ to
determine in constant time if $u$ is connected to other
deactivated vertices from \Voff{} via connected components.
This avoids having to check all connected components $C_i$ of the vertex $u$ which
could take time $\Theta(\non)$.

The preprocessing takes time $O(\mon)$ to compute the connected components $C_i$
and labeling the vertices in \Von.
Using one pass over all edges we can compute the arrays containing
the connectivity information between the $C_i$ and \Voff, i.e., we can fill the arrays
$A_{C_i}$ and the first $k$ components of the $A_u$. This takes $O(m)$ time.

Notice that if $u$ is adjacent to $C_i$, then
$A_u$ must have a 1 wherever $A_{C_i}$ has a 1.
Then to finish building the arrays $A_u$,
we can compute the last $\noff$ entries of $A_u$ as the bitwise OR of the arrays $A_{C_i}$
for all $C_i$ which $u$ is adjacent to. This can be done in time
$O(k \cdot \noff) = O(\non \cdot \noff)$ for a single vertex $u$.
Since we have \noff{} vertices in \Voff, computing all $A_u$
takes time $O(\noff^2 \cdot \non)$. The computation of the $A_u$ therefore dominates the
running time of the preprocessing.

The space we require during the preprocessing is $O(\noff)$ for each
connected component of \Gon{} and $O(n)$ for each vertex in \Voff. Hence, in total
we require $O(\noff \cdot n)$ space and preprocessing time $O(\noff^2 \cdot \non + m)$.

\subsection{Updates}
During an update which activates $d$ vertices from a set $I$,
we build the \emph{\supergraph}~$S$ with the vertices of $I$ as its nodes.
We add an edge between
a pair of vertices $u, v \in I$ if they are connected by a connected component $C$ of \Gon{}.
Notice that the \supergraph{} encodes the connectivity of the vertices
in $I$ via the connected components of \Gon{}.

Computationally, this can be done in time $O(d^2)$:
For each pair of vertices $u, v \in I$, we check in $A_u$ if $u$ is
connected to $v$ via a connected component in time $O(1)$.
As we have to consider $O(d^2)$ pairs of vertices, the total
time to construct the \supergraph{} is $O(d^2)$.

Finally, we compute the connected components $S_1, \dots, S_\ell$
of $S$ and label each vertex in $S$ with its connected component.
This can be done in time $O(|S|) = O(d^2)$.
Hence, the total update time is $O(d^2)$.

\subsection{Queries}
Consider a query if two activated vertices $u$ and $v$
are connected.

We find the connected components $C_u$ and $C_v$ of $u$ and $v$, respectively.
If $C_u = C_v$, then we return that $u$ and $v$ are connected and we are done.

Otherwise, let $S_i$ be a connected component of $S$.
We consider each vertex $w$ of $S_i$ and check if it is connected to
$C_u$ or $C_v$ using $A_{C_u}$ and $A_{C_v}$. After considering all vertices of $S_i$,
we check if both $C_u$ and $C_v$ are connected to $S_i$. If this is the case,
we return that $u$ and $v$ are connected, otherwise, we proceed to the next
connected component of $S$.

During the query we considered each vertex in $S$ exactly once and
spent time $O(1)$ processing it. Hence, the total query time is $O(d)$.

It is left to prove the correctness of the result of the queries.
This is done in the following lemma.

\begin{lemma}
  Consider an update which activates the vertices from a set $I \subseteq \Voff$.
  Then a query if two vertices $u$ and $v$ are connected
  in $G$ after the update delivers the correct result.
\end{lemma}
\begin{proof}
If in the query procedure we encountered that
$C_u = C_v$, then the result of the algorithm is clearly correct.

If $C_u \neq C_v$, then observe that the algorithm returns true if and only
if $C_u$ and $C_v$ are connected by the set $I$:
Let $S_i$ be the connected component of $S$ for which the query returns true.
Then there must exist vertices $w_1, \dots, w_t$ in the \supergraph{} such that
(1) $w_1$ is adjacent to $C_u$, (2) $w_t$ is adjacent to $C_v$
and (3) $(w_i, w_{i+1})$ is an edge in $S$ for all $i = 1,\dots,t-1$.
The first two claims are true because the query procedure checks this
in the arrays $A_{C_u}$ and $A_{C_v}$.
By construction of the \supergraph{}, the \supergraph{} has an edge $(w_i, w_{i+1})$
if and only if those vertices are connected by a connected component (this follows from
what we preprocessed in the arrays $A_{w_i}$).
This implies that a query returns true iff $C_u$ and $C_v$ are connected by the set $I$.

By Lemma~\ref{lem:paths} the algorithm returns the correct answer.
\end{proof}

\section{Fully Dynamic Algorithm}
\label{Sec:FullyDynamicAlgo}

In this section, we present the main result of the paper.
We provide a data structure for the fully dynamic subgraph connectivity problem
with sensitivity $d$, i.e., we process a batch update which changes the states
of at most $d$ vertices.
Our algorithm uses a data structure for the decremental problem as a subprocedure.
Assume the decremental algorithm uses space~$S$, preprocessing time~$t_p$,
update time~$t_u$ and query time~$t_q$.
Then the fully dynamic algorithm uses space~$O(\noff^2 \cdot S)$,
preprocessing time~$O(\noff^2 \cdot t_p)$, update time~$O(d^2 \cdot \max\{t_u,t_q\})$
and query time~$O(d \cdot t_q)$.

We reuse the \supergraph s which we used in the incremental algorithm.
For the construction of the \supergraph s
we replace the vectors $A_u$ and $A_{C_i}$ of the previous section by
slightly augmented versions of \Gon{} which are equipped with a
decremental subgraph connectivity data structure, e.g., the
one of Lemma~\ref{Lem:DecAlgo} by Duan and Pettie~\cite{duan2010connectivity}.
The purpose of the augmented graphs is to check if a pair of initially deactivated vertices 
is connected via a connected component after deactivating some vertices of \Von.

We sketch the main steps of our algorithm.
In the preprocessing we build an augmented graph for each pair of vertices of \Voff;
each augmented graph is equipped with a decremental subgraph connectivity data structure.
In an update, we first process the vertex deactivations in the augmented graphs.
Then we build the \supergraph{} of vertices that were activated.
Queries are handled similarly to the incremental algorithm
by using the \supergraph{}, but we have to check if the vertices
of the \supergraph{} can still reach the query vertices
(this connectivity may have been destroyed by the vertex deactivations).

\subsection{Preprocessing}
For each pair of nodes $u,v \in \Voff$, we build the augmented graph
$G_{u,v} = G[\Von \cup \{u,v\}]$, i.e., $G_{u,v}$ consists of \Gon{} after
adding $u$ and $v$. Observe that $u$ and $v$ cannot introduce more than
$O(\non)$ edges and hence $G_{u,v}$ still has $O(\non)$ vertices and
$O(\mon)$ edges.
We equip $G_{u,v}$ with a decremental subgraph connectivity
data structure with sensitivity $d$.
Later, we use the graph $G_{u,v}$ to check if $u$ and $v$ are
connected via a connected component after deleting vertices from \Gon;
intuitively, the graphs $G_{u,v}$ replace the
vectors $A_u$ and $A_v$ of the incremental algorithm.
We need space $O(\noff^2 \cdot S)$ to store the $G_{u,v}$ where
$S$ is the space to store \Gon{} with the decremental data structure.

For each $u \in \Voff$, we build the graph $G_u = G[\Von \cup \{u\}]$
and equip it with the decremental data structure;
we further equip \Gon{} with the decremental data structure.
We use the graphs $G_u$ to replace the arrays $A_{C_i}$ of the incremental
algorithm; we cannot use the arrays anymore because the connected components
of \Gon{} can fall apart due to vertex deactivations.
The space we need to store the graphs $G_u$ and $G$ is $O(\noff \cdot S)$.

In total, the preprocessing takes space
$O(\noff^2 \cdot S)$ and time $O(\noff^2 \cdot t_p)$.

\subsection{Updates}
Assume that we get an update $U$ which deactivates the vertices of a set $D \subseteq \Von$
and activates the vertices of a set $I \subseteq \Voff$ with $|D| + |I| \leq d$.
Our update procedure has two steps:
We first remove the vertices in $D$ from $G_{u,v}$ for all newly activated
vertices $u,v \in I$.
After that we build the \supergraph{} consisting of the vertices
of $I$ as we did in the incremental algorithm.

We describe the sketched steps of the update procedure in more detail.
Firstly, we process the deletions of the set $D$.
For each pair $u,v \in I$, we delete the vertices of $D$ in $G_{u,v}$ in time $t_u$.
Since we have $O(d^2)$ pairs of vertices of $I$ to consider,
this takes time $O(d^2 \cdot t_u)$.

We update $\Gon$ and all $G_u$ by deleting the vertices the vertices from $D$.
This does not take longer than updating the graphs $G_{u,v}$.

Secondly, we build the \supergraph{} consisting of the vertices in $I$.
For each pair of vertices $u,v \in I$, we add an edge $e = (u,v)$ to the \supergraph{}
if a query in $G_{u,v}$ returns that $u$ and $v$ are connected.
Such a query takes time $t_q$.
The time we spend to build the \supergraph{} is $O(d^2 \cdot t_u)$.
Finally, we compute the connected components of the \supergraph{} in time $O(d^2)$.

Altogether, the total update time of the update procedure is
$O(d^2 \cdot \max\{t_u,t_q\})$.

\subsection{Queries}
We handle the query if two vertices $u$ and $v$ of $G$ are connected
similarly as in the incremental algorithm by using the \supergraph.

Before we use the \supergraph{}, we query if $u$ and $v$
are connected in the instance of \Gon{} in which we deactivated the
vertices of the set $D$. If the query returns true, then $u$ and
$v$ are connected, otherwise, we proceed by using the \supergraph.

For each connected component $B$ of the \supergraph, we consider
each vertex $w \in B$ and we query in $G_w$ if $w$ is connected to $u$ or $v$.
If $B$ had vertices $w, w'$ which are connected to $u$ and $v$, respectively,
then we return that $u$ and $v$ are connected.
Otherwise, we proceed to the next connected component of the \supergraph.

The total query time of our algorithm is $O(d \cdot t_q)$ as in the worst case
we have to perform a query in $G_w$ for each of the $O(d)$ vertices $w \in I$.

Notice that due to the vertex deactivations we cannot precompute
the connected components $C_i$ of \Gon{} and their connectivity with vertices in \Voff{}
as we did in the incremental algorithm:
Each $C_i$ may consist of $\Theta(\non)$ vertices and might as well fall apart
into $\Theta(n)$ connected components after the vertex deactivations. Hence, in the
update procedure we cannot keep the information about the connectivity of the vertices
$C_i$ and the added vertices up to date, as this may take time $\Theta(n)$.
For our construction this also rules out obtaining a better query time.

We conclude the section by proving that the query returns the correct results
in the following lemma.

\begin{lemma}
  Consider an update $U$ deactivating the vertices from a set $D$ and activating
  the ones from a set $I$. Then a query if two vertices $u$ and $v$ are connected
  in $G$ after the update delivers the correct result.
\end{lemma}
\begin{proof}
In the query procedure, we first check if $u$ and $v$ are connected
in $\Gon$ after deleting the vertices from $D$.
Clearly, if the algorithm returns true,
then $u$ and $v$ are connected.

We move on to argue about the correctness in the case that $u$ and $v$
are not connected in the graph $H = \Gon \setminus D$.
Let $C_1, \dots, C_k$ be the connected components of $H$
(not those of \Gon) and let $C_u$ and $C_v$ be the
connected components of $u$ and $v$.
We show that a query returns that $u$ and $v$ are connected
if and only if $C_u$ and $C_v$ are connected by the set $I$.
Then Lemma~\ref{lem:paths} implies the correctness of the algorithm.

Observe that a query returns that $u$ and $v$ are connected if and only if
there exists a connected component $B$ in the \supergraph{}
which contains vertices $w_1, \dots, w_\ell \in B \subseteq I$, such that
(1) $w_1$ is connected to $u$,
(2) $w_\ell$ is connected to $v$ and
(3) there is an edge between $w_i$ and $w_{i+1}$ in the \supergraph{}
    for all $i = 1,\dots,\ell-1$:
We obtain the first two properties from the queries in $G_{w_1}$
and $G_{w_\ell}$; 
the third property is true due to the queries in the augmented graphs $G_{w_i, w_{i+1}}$
and implies that the $w_i$ are connected via connected components.

Hence, we conclude that a query returns that $u$ and $v$ are connected
if and only if $C_u$ and $C_v$ are connected by the set $I$.
Lemma~\ref{lem:paths} implies that the algorithm is correct.
\end{proof}

{\bf Acknowledgements.}
We would like to thank the reviewers for their helpful comments.
The research leading to these results has received funding from the European Research Council under the
European Union's Seventh Framework Programme (FP/2007-2013) / ERC Grant Agreement
no.\ 340506.

\bibliography{p135-henzinger}{}
\bibliographystyle{plainurl}

\end{document}